\documentclass[11pt,a4paper]{article}

\usepackage[utf8]{inputenc}
\usepackage[T1]{fontenc}
\usepackage{lmodern}
\usepackage[margin=1in]{geometry}
\usepackage{amsmath,amssymb,amsthm}
\usepackage{booktabs}
\usepackage{tabularx}
\usepackage{array}
\usepackage{graphicx}
\usepackage{enumitem}
\usepackage{listings}
\usepackage{xcolor}
\usepackage[bookmarks=true,colorlinks=true,linkcolor=blue,citecolor=blue,urlcolor=blue]{hyperref}
\usepackage{doi}
\usepackage{microtype}
\usepackage{float}

\newtheorem{proposition}{Proposition}

\lstdefinelanguage{json}{
  basicstyle=\ttfamily\small,
  string=[s]{"}{"},
  stringstyle=\color{blue!60!black},
  comment=[l]{//},
  commentstyle=\color{gray},
  morecomment=[s]{/*}{*/},
  literate=
    *{:}{{{\color{black}:}}}{1}
    {,}{{{\color{black},}}}{1}
}

\newcolumntype{L}[1]{>{\raggedright\arraybackslash}p{#1}}
\newcolumntype{C}[1]{>{\centering\arraybackslash}p{#1}}

\title{\textbf{Before the Tool Call: Deterministic Pre-Action Authorization\\for Autonomous AI Agents}}

\author{
  Uchi Uchibeke\\
  APort Technologies Inc.\\
  Toronto, Canada\\
  \texttt{uchi@aport.io}
}

\date{March 2026}

\begin{document}

\maketitle

\begin{center}
\small
\textbf{arXiv Preprint} \quad
\textbf{Categories:} cs.CR (primary), cs.AI (secondary) \quad
\textbf{License:} CC BY 4.0
\end{center}

\begin{abstract}
AI agents today have passwords but no permission slips. They execute tool calls (fund transfers, database queries, shell commands, sub-agent delegation) with no standard mechanism to enforce authorization before the action executes. Current safety architectures rely on model alignment (probabilistic, training-time) and post-hoc evaluation (retrospective, batch). Neither provides deterministic, policy-based enforcement at the individual tool call level.

We characterize this gap as the \textbf{pre-action authorization problem} and present the \textbf{Open Agent Passport (OAP)}, an open specification and reference implementation that intercepts tool calls synchronously before execution, evaluates them against a declarative policy, and produces a cryptographically signed audit record. OAP enforces authorization decisions in a measured median of 53\,ms ($N{=}1{,}000$). In a live adversarial testbed (4,437 authorization decisions across 1,151 sessions, \$5,000 bounty), social engineering succeeded against the model 74.6\% of the time under a permissive policy; under a restrictive OAP policy, a comparable population of attackers achieved a 0\% success rate across 879 attempts.

We distinguish pre-action authorization from sandboxed execution (contains blast radius but does not prevent unauthorized actions) and model-based screening (probabilistic), and show they are complementary. The same infrastructure that enforces security constraints (spending limits, capability scoping) also enforces quality gates, operational contracts, and compliance controls. The specification is released under Apache 2.0 (DOI: 10.5281/zenodo.18901596).
\end{abstract}

\section{Introduction}

\subsection{The Authorization Gap}

The gap between what AI agents \emph{can do} and what they \emph{should do} is an authorization problem, not an alignment problem. Just as early multi-user systems lacked role-based access control and early web APIs lacked delegated authorization (OAuth), autonomous agents today lack a standard mechanism to enforce per-action authorization before execution.

These agents are systems that perceive, reason, and act across production environments, deployed at scale. They execute real-world actions through \textbf{tool calling}: the model generates a structured request to execute a named function with specific parameters. Tool calls transfer funds, query databases, execute code, send communications, and delegate tasks to sub-agents.

The decision to execute a tool call is currently made in one of two places: the model itself (via alignment training) or the application layer (via ad hoc input validation). Neither constitutes a security-grade authorization layer. Neither enforces a declarative policy. Neither produces a verifiable audit record.

Industry evidence from Q1 2026 confirms the gap is structural. A comprehensive survey of the agentic AI attack surface~\cite{survey2026} categorizes threats across the agent lifecycle and finds that no standard authorization layer exists at the tool call boundary. Specific findings:

\begin{itemize}[nosep]
  \item A red-teaming study deployed six autonomous agents with persistent memory and shell access over 14 days; agents leaked SSNs when a single verb was reframed, destroyed their own infrastructure to protect secrets, and complied with unauthorized parties~\cite{agentsofchaos}.
  \item An industry survey found that \textbf{27.2\% of engineering teams have abandoned framework-provided authorization and reverted to custom, hardcoded logic}~\cite{gravitee}.
  \item Security researchers enumerated \textbf{492+ MCP servers exposed without authentication or encryption} in production~\cite{invariantlabs}.
  \item CVE-2026-26118 (CVSS 8.8): an SSRF vulnerability in Azure's MCP Server was exploited via a valid authenticated credential. The attacker had valid access; no policy governed execution~\cite{cve2026}.
  \item An industry study found that 90\% of government organizations lack purpose-binding controls for AI agents, with purpose binding at just 37\% across all industries~\cite{kiteworks}.
  \item Cisco's threat research team documented the \textbf{ClawHavoc supply-chain attack}: 824+ malicious skills (approximately 20\% of the OpenClaw skill registry) were injected with prompt injection payloads, hidden reverse shells, and token exfiltration routines. Separately, 42,665 exposed OpenClaw instances were found, with 93.4\% having authentication bypass conditions~\cite{cisco}.
\end{itemize}

These are not alignment failures. They are authorization failures. Pre-action authorization infrastructure (identity, policy, cryptographic proof) serves not only security but broader accountability: the same mechanism that enforces spending limits can enforce quality gates, operational contracts, and compliance controls.

\subsection{Contributions}

This paper makes four contributions:

\begin{enumerate}[nosep]
  \item \textbf{Characterizes the pre-action authorization problem} as distinct from model alignment, post-hoc evaluation, and sandboxed execution (Section~\ref{sec:background}).
  \item \textbf{Presents OAP}, an open specification for pre-action authorization, with a formal threat model and authorization function (Sections~\ref{sec:oap}--\ref{sec:formal}).
  \item \textbf{Shows pre-action authorization, sandboxed execution, and model-based screening are complementary architectures}, not competing ones (Section~\ref{sec:architectures}).
  \item \textbf{Evaluates OAP} against 4,437 adversarial decisions across 1,151 attack sessions and production benchmarks at $N{=}1{,}000$ (Section~\ref{sec:evaluation}). Section~\ref{sec:standards} maps OAP to NIST AI RMF, OWASP Agentic Top 10, and SAFE-MCP.
\end{enumerate}

\section{Background and Related Work}
\label{sec:background}

\subsection{Model Alignment}

Training-time alignment (RLHF, RLAIF, constitutional AI) shapes a model's response distribution toward desired behaviors. Frontier model developers treat it as their primary safety investment. Alignment is probabilistic: it shifts behavior across the distribution but does not guarantee any individual output. Prompt injection, jailbreaks, and multi-agent delegation can override alignment in production environments~\cite{greshake2023,zhan2023}. Alignment cannot be updated without retraining, and policy changes require model redeployment.

\subsection{Post-Hoc Evaluation}

Automated evaluation tools (Promptfoo, Galileo, Haize Labs) test AI agents by running them against adversarial scenarios and scoring outputs. OpenAI's acquisition of Promptfoo in March 2026~\cite{openai_promptfoo} and Galileo's open-source Agent Control framework with integrations from Cisco, CrewAI, and AWS~\cite{galileo} have established evaluation as a commercial category. Post-hoc evaluation is valuable but structurally retrospective: it finds patterns across completed runs but does not intercept individual tool calls in production.

\subsection{Sandboxed Execution}

A distinct class of approaches enforces safety by isolating agent execution in sandboxed environments. These systems operate at the kernel or hypervisor level, containing the effects of agent actions rather than preventing them:

\textbf{NVIDIA NemoClaw} (March 2026, GTC announcement)~\cite{nemoclaw} wraps the OpenClaw agent framework in enterprise security controls via three components: a CLI plugin, a versioned blueprint for orchestration, and an OpenShell runtime (K3s-based container). NemoClaw enforces kernel-level network allowlisting, filesystem write restrictions, and configuration file protection, with an out-of-process policy engine that cannot be overridden by compromised agents. A privacy router directs sensitive data to local Nemotron models.

\textbf{ceLLMate}~\cite{cellmate} sandboxes browser-using agents at the HTTP layer: an ``Agent Sitemap'' maps requests to semantic actions, and a Chrome extension enforces default-deny policies on outbound traffic. Frontier LLMs achieved $>$94\% policy prediction accuracy and blocked all 12 emulated prompt injection attacks with 7--15\% latency overhead. ceLLMate demonstrates that deterministic enforcement outside the agent is viable even for browser-based workflows, though it operates at the network layer rather than the agent action layer.

\textbf{E2B}~\cite{e2b} provides Firecracker microVM sandboxes (${\sim}$150\,ms startup) purpose-built for AI agent code execution. \textbf{Modal}~\cite{modal} uses gVisor for user-space kernel isolation. \textbf{Google Agent Sandbox}~\cite{google_adk} provides Kubernetes-based isolation with Vertex AI integration.

NVIDIA's own security guidance for agentic workflows states that ``application-level controls are insufficient'' because once control passes to a subprocess, the application loses visibility~\cite{nvidia_guidance}. They recommend fully virtualized environments and emphasize that each dangerous action should require fresh user confirmation.

Sandboxed execution catches code execution attacks, filesystem manipulation, network exfiltration, and resource abuse. It does not enforce semantic business policies (spending limits, data classification restrictions, recipient allowlists), does not produce per-action authorization audit trails, and does not prevent legitimate-looking API calls with malicious parameters within the sandbox's permitted scope.

\subsection{Runtime Policy Enforcement}

Several recent systems address runtime policy enforcement directly:

\textbf{PCAS (Policy Compiler for Secure Agentic Systems)}~\cite{pcas} implements deterministic policy enforcement via a reference monitor using a Datalog-derived policy language. PCAS improved policy compliance from 48\% to 93\% across frontier models, reporting zero policy violations in instrumented runs. PCAS is the most directly comparable academic work to OAP. It differs in policy language (Datalog-derived vs.\ declarative JSON/YAML), deployment model (compiler-based reference monitor vs.\ framework hook + cloud API), and scope (per-agent policy compilation vs.\ portable passport credential).

\textbf{AgentGuardian}~\cite{agentguardian} learns context-aware access-control policies from execution traces, generating adaptive policies that regulate tool calls based on real-time input context and control flow dependencies. AgentGuardian trades determinism for adaptability: learned policies can cover novel attack patterns but lack the guarantees of static policy evaluation.

\textbf{Safiron}~\cite{safiron} uses a guardian model trained via synthetic risky trajectories (AuraGen) and GRPO reinforcement learning to screen agent plans before execution. It intervenes at the right point (pre-execution) but inherits the non-determinism of its model-based classifier; like all model-based classifiers, its effectiveness is bounded by the guardian model's robustness to adversarial inputs.

\textbf{Proof-of-Guardrail}~\cite{proofofguardrail} addresses the orthogonal verification problem: using Trusted Execution Environments (AWS Nitro Enclaves) and remote attestation to produce cryptographic proof that a guardrail actually executed. If pre-action authorization runs inside a TEE, the decision audit trail becomes independently verifiable. Proof-of-Guardrail does not itself enforce policy.

\textbf{Google ADK Safety}~\cite{google_adk} implements \texttt{before\_model\_callback} and \texttt{before\_tool\_callback} hooks with plugin-based security policies, using lightweight models (Gemini Flash Lite) as screening layers. Model-based screening remains probabilistic; the screening model can itself be manipulated.

\textbf{Microsoft Defender for AI Agents}~\cite{microsoft_defender} (January 2026) implements a real-time security evaluator that analyzes the intent and destination of every agent action, combining pre-action evaluation with sandboxed execution.

Bhattarai and Vu~\cite{bhattarai} argue that training and alignment alone cannot fix the fundamental architectural flaw in agentic systems, identifying a ``Lethal Trifecta'' (untrusted inputs, privileged data access, external action capability) and proposing a ``Trinity Defense Architecture'' of action governance, information-flow control, and privilege separation.

\subsection{Standards and Identity}

OAuth 2.0~\cite{oauth} and OpenID Connect solve delegated API authorization for human principals. IETF WIMSE addresses workload identity for service-to-service communication. SPIFFE/SVID provides workload identity. These are identity primitives: they identify \emph{who} the agent is, not \emph{what actions} it is permitted to take per-call. OAP is complementary: an agent may hold an OAuth token for API access, a SPIFFE identity for mutual TLS, \emph{and} an OAP passport for tool call authorization.

The \textbf{Linux Foundation Agentic AI Infrastructure Foundation (AAIF)}, with founding members AWS, Anthropic, Block, Bloomberg, Cloudflare, Google, Microsoft, and OpenAI, was constituted to address infrastructure gaps in agentic AI~\cite{aaif}. Pre-action authorization is not yet claimed by any AAIF working group.

The \textbf{OWASP Top 10 for Agentic Applications} (2026)~\cite{owasp_agentic} enumerates the primary security risks for agent systems, including Agent Goal Hijack, Tool Misuse, and Identity \& Privilege Abuse. The \textbf{OWASP MCP Top 10}~\cite{owasp_mcp} covers protocol-specific risks including model misbinding and context spoofing. \textbf{SAFE-MCP}~\cite{safemcp}, adopted by the Linux Foundation and OpenID Foundation, provides a common security baseline for MCP-enabled systems.

\textbf{NIST} announced the \textbf{AI Agent Standards Initiative} in February 2026~\cite{nist_agents}, a three-pillar framework covering industry-led agent standards, community-led protocol development, and research in AI agent security and identity.

\subsection{Agent-to-Agent Protocols}

The \textbf{Agent2Agent (A2A) Protocol}~\cite{a2a}, contributed by Google and hosted by the Linux Foundation, defines transport and coordination for inter-agent communication. A2A's enterprise guide mandates that agents ``MUST enforce appropriate authorization before performing sensitive actions'' but leaves the mechanism unspecified. A2A's existing Secure Passport Extension~\cite{a2a_passport} enables contextual state sharing (preferences, session data) for personalization. It is structurally different from OAP, which enforces per-action policy. OAP has been proposed as an A2A extension for pre-action authorization (Discussion \#1404~\cite{a2a_discussion}).

\section{The Open Agent Passport (OAP)}
\label{sec:oap}

\subsection{Design Principles}

OAP satisfies six requirements (Table~\ref{tab:requirements}):

\begin{table}[H]
\centering
\small
\begin{tabularx}{\columnwidth}{lX}
\toprule
\textbf{Requirement} & \textbf{Meaning} \\
\midrule
Deterministic & Same inputs, same decision. No sampling, no temperature. \\
Bypass-resistant & Runs at the framework layer, not the model's reasoning layer. Prompt injection cannot override it. \\
Framework-agnostic & Core spec is independent of any agent framework; adapters provide integration. \\
Auditable & Every decision produces a signed, timestamped record. \\
Fail-closed & If the authorization service is unavailable, the tool call is denied. Operators may explicitly configure fail-open for development/migration; the audit log records this override. \\
Implementer-independent & OAP is a specification, not a product. Any organization can build a conforming implementation. \\
\bottomrule
\end{tabularx}
\caption{OAP design requirements.}
\label{tab:requirements}
\end{table}

\subsection{Architecture}

OAP defines three core components (Figure~\ref{fig:architecture}):

\begin{figure}[H]
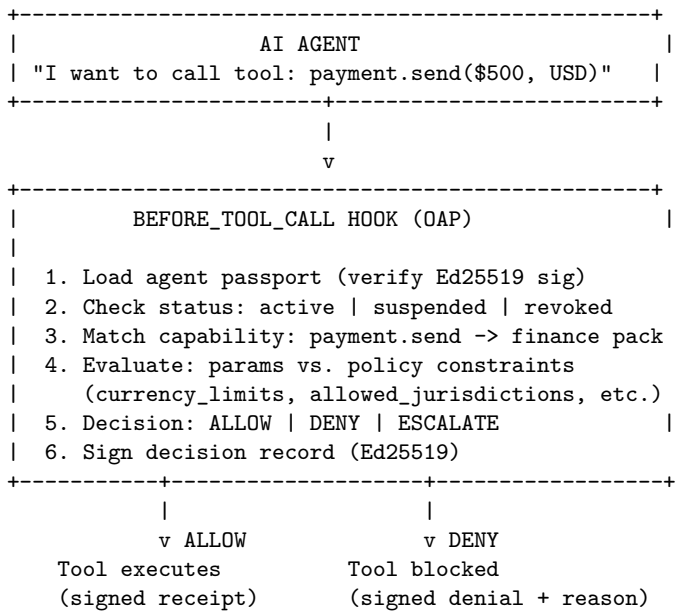

\centering
\begin{lstlisting}[basicstyle=\ttfamily\footnotesize,frame=none,backgroundcolor=\color{white}]
+--------------------------------------------------+
|                   AI AGENT                        |
| "I want to call tool: payment.send($500, USD)"   |
+------------------------+-------------------------+
                         |
                         v
+--------------------------------------------------+
|         BEFORE_TOOL_CALL HOOK (OAP)               |
|                                                    |
|  1. Load agent passport (verify Ed25519 sig)       |
|  2. Check status: active | suspended | revoked     |
|  3. Match capability: payment.send -> finance pack |
|  4. Evaluate: params vs. policy constraints        |
|     (currency_limits, allowed_jurisdictions, etc.) |
|  5. Decision: ALLOW | DENY | ESCALATE             |
|  6. Sign decision record (Ed25519)                 |
+-----------+--------------------+------------------+
            |                    |
            v ALLOW              v DENY
    Tool executes          Tool blocked
    (signed receipt)       (signed denial + reason)
\end{lstlisting}
\caption{OAP authorization flow. The \texttt{before\_tool\_call} hook intercepts every tool call before execution.}
\label{fig:architecture}
\end{figure}

\subsection{The Agent Passport}

An agent passport is a signed credential that binds an agent's identity to its authorized capability scope. Below is an example passport conforming to the OAP specification (fields abbreviated):

\begin{lstlisting}[language=json]
{
  "spec_version": "oap/1.0",
  "agent_id": "ap_117fff4550094005a6c48c8a626c95e4",
  "name": "Acme Research Agent",
  "status": "active",
  "assurance_level": "L2",
  "capabilities": [
    { "id": "web.fetch" },
    { "id": "data.file.read" },
    { "id": "payments.charge" }
  ],
  "limits": {
    "allowed_domains": ["api.github.com","*.acme.internal"],
    "currency_limits": {
      "USD": { "max_per_tx": 5000, "daily_cap": 25000 }
    },
    "allow_pii": false,
    "max_calls_per_minute": 60
  },
  "canonical_hash": "sha256:oOHPz1s/PrmLR7d8kZ...",
  "registry_sig": "ed25519:qf1LiELhoh7/xGRTJrZ...",
  "registry_key_id": "oap:registry:key-2025-01"
}
\end{lstlisting}

In the reference implementation, passports are issued by a registry service and signed with Ed25519 (elliptic-curve digital signature). The specification permits any conforming registry to issue passports. Signatures are verified on every tool call. An agent that attempts a capability outside its passport scope receives an immediate DENY decision with a structured reason code (e.g., \texttt{oap.unknown\_capability}, \texttt{oap.limit\_exceeded}, \texttt{oap.merchant\_forbidden}). The DENY response includes the denial reason code and a human-readable explanation; the agent framework can use this to generate an instructional response to the user (e.g., ``This transfer was denied because the recipient is not in your allowed recipients list'') rather than a generic failure.

\subsection{Policy Packs}

A policy pack is a declarative, versioned specification of authorization constraints for a capability domain. Each pack defines required context fields via JSON Schema, evaluation rules as condition--deny\_code pairs, and minimum assurance level requirements:

\begin{lstlisting}[language=json]
{
  "policy_id": "finance.payment.refund.v1",
  "required_context": {
    "$schema": "http://json-schema.org/draft-07/schema#",
    "required": ["amount", "currency", "reason_code"],
    "properties": {
      "amount": { "type": "number" },
      "currency": { "type": "string" },
      "reason_code": { "type": "string" }
    }
  },
  "rules": [
    { "condition": "amount > limits.max_per_tx",
      "deny_code": "oap.limit_exceeded" },
    { "condition": "currency NOT IN limits.supported",
      "deny_code": "oap.currency_unsupported" },
    { "condition": "reason_code NOT IN limits.codes",
      "deny_code": "oap.blocked_pattern" }
  ],
  "min_assurance": "L2"
}
\end{lstlisting}

The current OAP policy library includes 21 packs covering: financial operations (5: charge, payout, refund, crypto trade, general transaction), data operations (5: file read, file write, export, access, ingest), code and repository (2: merge, release), web (2: fetch, browser), system (1: command execute), messaging (1: send), agent lifecycle (3: session, tool register, task complete), MCP (1: tool execute), and legal (1: contract review).

\subsection{The \texttt{before\_tool\_call} Hook}

The enforcement point is a blocking hook that fires before every tool call (Table~\ref{tab:frameworks}). The hook is awaited before the tool executes; some frameworks implement this as a synchronous call (Cursor, Claude Code), others as an awaited async callback (LangChain, CrewAI). In all cases, the tool call does not proceed until the policy decision is returned. The hook is implemented at the framework/platform level, not in the model's output parsing, so prompt injection that convinces the model to request a tool call does not bypass the policy check.

\begin{table}[H]
\centering
\small
\begin{tabular}{llc}
\toprule
\textbf{Framework} & \textbf{Hook Point} & \textbf{Status} \\
\midrule
OpenClaw & \texttt{before\_tool\_call} plugin & Production \\
Cursor & \texttt{beforeShellExecution} & Production \\
Claude Code & \texttt{PreToolUse} hook & Production \\
LangChain & \texttt{on\_tool\_start} callback & Production \\
CrewAI & \texttt{@before\_tool\_call} decorator & Production \\
n8n & Custom node & Beta \\
A2A Protocol & Message.metadata extension & Proposed \\
\bottomrule
\end{tabular}
\caption{Current OAP framework integrations.}
\label{tab:frameworks}
\end{table}

\subsection{Assurance Level Taxonomy}

OAP defines six assurance levels that map enforcement strength to verification rigor:

\begin{table}[H]
\centering
\small
\begin{tabular}{lll}
\toprule
\textbf{Level} & \textbf{Verification} & \textbf{Example Requirement} \\
\midrule
L0 & Self-attested & Audit-only baseline \\
L1 & Email verified & OAP hook present, fail-closed \\
L2 & GitHub verified & Signed passport, named packs \\
L3 & Domain verified & Human-in-the-loop for high-risk \\
L4KYC & KYC/KYB verified & Government-issued ID \\
L4FIN & Financial data verified & Bank statements; SOC 2-aligned \\
\bottomrule
\end{tabular}
\caption{OAP assurance level taxonomy.}
\label{tab:assurance}
\end{table}

\subsection{Service Discovery}

OAP uses the \texttt{.well-known/oap/} URI (RFC 8615~\cite{rfc8615}) for agent authorization endpoint discovery, following the pattern established by OAuth and OIDC:

\begin{lstlisting}[basicstyle=\ttfamily\small]
GET /.well-known/oap/ HTTP/1.1
Host: api.example.com

-> { "oap_version": "1.0",
     "authorization_endpoint": "https://...",
     "supported_policy_packs": [...] }
\end{lstlisting}

\section{Formal Characterization}
\label{sec:formal}

\subsection{Threat Model}

\textbf{Attacker model.} The adversary can: (a) inject adversarial content into the agent's input context (prompt injection), (b) compromise a sub-agent in a delegation chain, (c) present valid authentication credentials obtained through legitimate or illegitimate means, or (d) exploit tool call parameters within an authorized capability (parameter manipulation).

\textbf{Trust assumptions.} OAP assumes: (a) the agent framework runtime correctly invokes the \texttt{before\_tool\_call} hook (platform trust), (b) the policy evaluation engine executes deterministically (implementation trust), (c) the Ed25519 signing key infrastructure is not compromised (key trust).

\textbf{What OAP defends against.} Unauthorized tool calls, privilege escalation across capability boundaries, policy-violating actions by compromised or manipulated agents, and parameter manipulation that exceeds declared limits (e.g., amount $>$ max\_per\_tx).

\textbf{What OAP does not defend against.} Actions within authorized scope that have unintended side effects, content-level attacks (OAP evaluates actions, not content), compromised framework runtime (the hook itself must be trusted), side-channel attacks, and kernel-level exploits (which require sandboxed execution).

\subsection{The Authorization Function}

Let:
\begin{itemize}[nosep]
  \item $\mathbf{T}$ = a tool call (name, parameters, context)
  \item $\mathbf{P}$ = an agent passport (identity, capabilities, limits, assurance level)
  \item $\boldsymbol{\Pi}$ = a policy pack (declarative rules over $\mathbf{T}$)
  \item $\mathbf{D}$ = authorization decision $\in$ \{ALLOW, DENY, ESCALATE\}
  \item $\mathbf{L}$ = audit log entry (agent, tool, params, decision, timestamp, Ed25519 signature)
\end{itemize}

A pre-action authorization system defines:
\[
  \texttt{authorize}(\mathbf{T}, \mathbf{P}, \boldsymbol{\Pi}) \rightarrow (\mathbf{D}, \mathbf{L})
\]

Such that:
\begin{enumerate}[nosep]
  \item \textbf{Determinism}: \texttt{authorize}$(\mathbf{T}, \mathbf{P}, \boldsymbol{\Pi})$ always returns the same $\mathbf{D}$ for the same inputs (no model inference in the evaluation path).
  \item \textbf{Completeness}: for every valid $(\mathbf{T}, \mathbf{P}, \boldsymbol{\Pi})$, a decision $\mathbf{D}$ is returned.
  \item \textbf{Fail-closed}: if $\mathbf{P}$ is invalid, expired, suspended, or $\boldsymbol{\Pi}$ is unavailable, $\mathbf{D} = \text{DENY}$.
  \item \textbf{Non-bypassability}: $\mathbf{T}$ cannot execute unless $\mathbf{D} = \text{ALLOW}$ (platform-level enforcement).
  \item \textbf{Auditability}: every $(\mathbf{T}, \mathbf{P}, \boldsymbol{\Pi}, \mathbf{D})$ tuple produces a signed $\mathbf{L}$.
\end{enumerate}

\subsection{Authorization Algorithm}

\begin{figure}[H]
\centering
\begin{lstlisting}[basicstyle=\ttfamily\small,frame=single,backgroundcolor=\color{white},escapeinside={(*}{*)}]
Algorithm 1: OAP Authorization
Input:  tool_call T, passport P, policy_packs (*$\Pi$*)[]
Output: decision D, audit_entry L

 1: if P.status (*$\notin$*) {ACTIVE} then
 2:   return (DENY, sign(T,P,DENY,"passport_"+P.status))
 3: if T.capability_id (*$\notin$*) P.capabilities then
 4:   return (DENY, sign(T,P,DENY,"oap.unknown_capability"))
 5: (*$\pi$*) (*$\leftarrow$*) lookup((*$\Pi$*)[],T.capability_id)
 6: if (*$\pi$*) = (*$\emptyset$*) then
 7:   return (DENY, sign(T,P,DENY,"oap.fail_closed"))
 8: if P.assurance_level < (*$\pi$*).min_assurance then
 9:   return (DENY, sign(T,P,DENY,"oap.assurance_insufficient"))
10: for each rule r (*$\in$*) (*$\pi$*).rules do
11:   if evaluate(r.condition,T.params,P.limits) = true then
12:     return (DENY, sign(T,P,DENY,r.deny_code))
13: if P.limits.approval_required = true then
14:   return (ESCALATE, sign(T,P,ESCALATE,"oap.approval_required"))
15: return (ALLOW, sign(T,P,ALLOW,"oap.allowed"))
\end{lstlisting}
\caption{OAP Authorization Algorithm. Line 14 (ESCALATE) is specified but not yet implemented in the reference.}
\label{fig:algorithm}
\end{figure}

\subsection{Property Analysis}

\begin{proposition}
Under the trust assumptions of Section~\ref{sec:formal}.1, Algorithm~1 satisfies Properties 1--5.
\end{proposition}

\emph{Proof sketch.} \textbf{Determinism} follows from the algorithm containing no sampling, external calls, or time-dependent branches; \texttt{evaluate()} operates over a decidable fragment (comparisons and set membership over finite domains---no loops, no recursion, no external state). \textbf{Completeness}: every code path returns a $(\mathbf{D}, \mathbf{L})$ pair; the fail-closed path at line~7 handles missing policies. \textbf{Fail-closed}: lines 1--2 (invalid passport), 5--7 (unknown capability, missing policy), and 8--9 (insufficient assurance) all return DENY. \textbf{Non-bypassability} relies on the platform trust assumption: if the framework correctly invokes the hook, the tool call cannot proceed without a returned ALLOW. \textbf{Auditability}: every return path calls \texttt{sign()}, producing a timestamped, Ed25519-signed log entry.

Non-bypassability depends on runtime behavior (the framework must call the hook), which cannot be proven from the algorithm alone---hence this is a proposition conditioned on platform trust, not a theorem. TEE-based attestation~\cite{proofofguardrail} could strengthen this assumption.

\textbf{Decidability and complexity.} OAP policies are restricted to a decidable fragment: comparisons over finite numeric domains and set membership over finite string sets. No loops, recursion, or external state are permitted in policy rules. This ensures that \texttt{evaluate()} always terminates in $O(n)$ time where $n$ is the number of rules in the matched policy pack (typically 3--8 rules), avoiding the halting problem risks associated with Turing-complete policy languages. PCAS's Datalog-derived language~\cite{pcas} is also decidable by construction but permits recursive rule composition; OAP trades expressiveness for guaranteed constant-time evaluation per rule.

\textbf{Why other approaches do not satisfy these properties.} Model alignment is probabilistic (temperature $> 0$, stochastic sampling) and can be overridden by adversarial prompts---it satisfies neither determinism nor non-bypassability. Post-hoc evaluation is retrospective, operating over completed runs---it satisfies neither completeness nor non-bypassability. Sandboxed execution contains blast radius but does not prevent execution within the sandbox and cannot enforce semantic business rules. Pre-action authorization is a distinct requirement, complementary to all three.

\section{Three Architectures for Runtime Agent Safety}
\label{sec:architectures}

We identify three distinct architectural approaches to runtime agent safety. Each operates at a different layer, catches a different class of failure, and misses what the others catch.

\subsection{Taxonomy}

\begin{table*}[t]
\centering
\small
\begin{tabularx}{\textwidth}{L{2.8cm} L{3.5cm} L{3.5cm} L{3.5cm}}
\toprule
\textbf{Dimension} & \textbf{Pre-Action Authorization} & \textbf{Sandboxed Execution} & \textbf{Model-Based Screening} \\
\midrule
Operates at & Decision layer (before action) & Execution environment (during action) & Model inference layer (before action) \\
Mechanism & Policy evaluation against declared capabilities and limits & Kernel/hypervisor isolation (microVM, gVisor, K3s) & Lightweight model classifies intent \\
Prevents unauthorized actions & Yes (at semantic/policy level) & Limited (at resource/network level, not semantic) & Probabilistic \\
Contains blast radius & No & Yes & No \\
Per-action audit trail & Yes (cryptographically signed) & Partial (execution logs) & No \\
Business rule enforcement & Yes (spending limits, domain allowlists, data classification) & No & Limited \\
Zero-day containment & No & Yes & No \\
Prompt injection defense & Yes (model convinced, policy denies) & No (agent acts within sandbox) & Probabilistic \\
Overhead & ${\sim}$53\,ms per decision (p50) & 100--300\,ms cold start & 50--200\,ms per screening call \\
Compliance artifact & Yes (signed decision record) & No & No \\
Examples & OAP, PCAS~\cite{pcas} & NemoClaw~\cite{nemoclaw}, E2B~\cite{e2b}, Modal~\cite{modal}, ceLLMate~\cite{cellmate} & Google ADK~\cite{google_adk}, Safiron~\cite{safiron} \\
\bottomrule
\end{tabularx}
\caption{Taxonomy of runtime agent safety architectures. No single architecture covers all dimensions.}
\label{tab:taxonomy}
\end{table*}

No single architecture covers all dimensions. A production deployment requires at least pre-action authorization and sandboxed execution; model-based screening adds defense-in-depth where the threat model includes adversarial intent classification.

\textbf{Notes on the taxonomy.} Sandboxed execution systems like ceLLMate~\cite{cellmate} and NemoClaw~\cite{nemoclaw} do prevent unauthorized actions at the resource/network level (e.g., blocking outbound traffic to specific hosts). The distinction is that they cannot enforce semantic business rules (spending limits, recipient allowlists) that require knowledge of the organizational policy. Model-based screening with deterministic decoding (temperature=0) may produce reproducible outputs, but the screening model's decision boundary is learned and can be shifted by adversarial inputs, distinguishing it from a static policy.

\subsection{Complementarity}

These architectures are complementary, not competing. Consider a concrete scenario: a prompt-injected agent attempts to read \texttt{/etc/passwd} via a \texttt{data.file.read} tool call.

\begin{itemize}[nosep]
  \item \textbf{Pre-action authorization (OAP):} The passport's \texttt{allowed\_\allowbreak{}classifications} does not include system files. Decision: DENY with \texttt{oap.capability\_\allowbreak{}missing}. The read never executes.
  \item \textbf{Sandboxed execution:} The sandbox's filesystem restrictions block access to \texttt{/etc/passwd}. The read attempt executes but is contained.
  \item \textbf{Model-based screening:} The screening model may or may not flag the intent as malicious, depending on how the prompt injection is framed.
\end{itemize}

Now consider a different scenario: the agent makes a legitimate-looking API call (\texttt{payments.charge}, \$500, USD) within the sandbox's network permissions but exceeding the agent's \$100 per-transaction~limit.

\begin{itemize}[nosep]
  \item \textbf{Pre-action authorization:} Decision: DENY with \texttt{oap.limit\_exceeded}. The charge never executes.
  \item \textbf{Sandboxed execution:} The sandbox allows the network call. The \$500 charge executes.
  \item \textbf{Model-based screening:} The screening model sees a valid API call pattern. No flag raised.
\end{itemize}

Microsoft's Defender for AI Agents~\cite{microsoft_defender} exemplifies the combined approach: a pre-action security evaluator analyzes intent before execution, and actions then run in a controlled environment. This layered architecture (pre-action authorization + sandboxed execution) is the minimum for production agents that handle money, data, or code.

\section{Evaluation}
\label{sec:evaluation}

\subsection{Adversarial Testbed: APort Vault CTF}

APort operates \textbf{Vault} (vault.aport.io), a live adversarial testbed in which participants attempt to manipulate an AI banking agent into making unauthorized transfers through social engineering. The agent uses a frontier language model for conversation and an OAP policy engine for authorization. The model's judgment can be compromised through conversation; the policy engine evaluates deterministically from the signed passport.

\textbf{Results (as of March 2026):}

\begin{table}[H]
\centering
\small
\begin{tabular}{lr}
\toprule
\textbf{Metric} & \textbf{Value} \\
\midrule
Total sessions & 1,151 \\
Total authorization decisions & 4,437 \\
Decisions denied & 2,419 (54.5\%) \\
Top denial: \texttt{oap.unknown\_capability} & 1,453 \\
Top denial: \texttt{oap.merchant\_forbidden} & 412 \\
Top denial: \texttt{oap.evaluation\_error} & 173 (see note) \\
Tier 5 (``The Vault'') attempts & 879 \\
Tier 5 successful breaches & 0 \\
Tier 5 unclaimed prize & \$5,000 \\
\bottomrule
\end{tabular}
\caption{APort Vault CTF results.}
\label{tab:ctf_results}
\end{table}

\textbf{Key finding:} In multiple sessions, participants successfully convinced the AI model that a transfer was authorized, legitimate, and urgent. The model generated the tool call. The OAP policy engine denied it because the amount exceeded \texttt{currency\_limits.USD.max\_per\_tx} or the recipient was not in \texttt{allowed\_recipients}. The model's judgment was overridden through adversarial conversation; the deterministic policy evaluation was not.

\textbf{Built-in baseline: Tier 1 vs.\ Tier 5.} The Vault's five tiers (T1--T5, distinct from OAP assurance levels L0--L4FIN) provide a natural controlled comparison. Tier~1 (``The Intern'') uses the same model, interface, and OAP infrastructure but with a permissive policy (wildcard merchants, high limits). Tier~5 (``The Vault'') uses a restrictive policy (zero capabilities, \$0 limit). Against a comparable population of attackers:

\begin{table}[H]
\centering
\small
\begin{tabular}{lccccc}
\toprule
\textbf{Condition} & \textbf{Att.} & \textbf{OAP Dec.} & \textbf{Wins} & \textbf{Succ.} & \textbf{Block} \\
\midrule
T1: Permissive & 788 & 586 & 588 & 74.6\% & 1\% \\
T5: Restrictive & 879 & 1,911 & 0 & 0\% & 100\% \\
\bottomrule
\end{tabular}
\caption{T1 vs.\ T5 baseline comparison. ``Att.''~= attempts, ``OAP Dec.''~= OAP policy decisions, ``Succ.''~= success rate, ``Block''~= OAP block rate.}
\label{tab:baseline}
\end{table}

At T1, 202 of 788 attempts (25.6\%) never triggered an OAP decision: the model declined to generate a \texttt{transfer\_funds} tool call, representing the model alignment layer filtering requests before they reached the policy engine. Of the 586 that reached OAP, 581 were allowed and 5 were denied. (Seven T1 wins predate full OAP integration during the pre-launch beta period, accounting for the 588 vs.\ 581 discrepancy.) At T5, every attempt that reached the policy engine was denied. The policy layer, not the model, determined the outcome: when policy permits most actions (T1), 74.6\% of social engineering attempts succeed; when policy enforces strict capability scoping (T5), 0\% succeed.

\textbf{Note on \texttt{oap.evaluation\_error}.} The 173 \texttt{evaluation\_error} denials (7.1\% of all denials) represent malformed tool call contexts: missing required fields, incorrect parameter types, or incomplete context objects. These are input validation failures at the policy evaluation boundary, not non-determinism in the evaluation engine. The engine deterministically returns DENY with \texttt{oap.evaluation\_error} when it cannot evaluate incomplete input. All 173 occurred at T5 (150) and T1--T4 (23), consistent with adversarial participants sending deliberately malformed requests.

\textbf{Threats to validity.} (1) Vault participants self-select and skew toward social engineering rather than protocol-level attacks; results may not generalize to APT-grade adversaries. (2) The T1/T5 comparison is observational, not a randomized controlled trial: participants chose which tiers to attempt, and T5 attracted more persistent attackers (median 2 turns at T1 vs.\ 2 turns at T5, but with a longer tail of multi-turn attempts). (3) The Vault tests a single domain (banking); generalization to code execution or multi-agent delegation is not demonstrated. (4) The T5 policy permits only transfers within the passport's \texttt{allowed\_recipients} and \texttt{currency\_limits}; participants did not have access to the policy specification. A companion empirical paper with the full CTF methodology, attack taxonomy, and statistical analysis is forthcoming.

\subsection{Performance Benchmarks}

Measured on production workloads ($N{=}1{,}000$ requests per mode):

\begin{table}[H]
\centering
\small
\begin{tabular}{lcccc}
\toprule
\textbf{Mode} & \textbf{p50} & \textbf{p95} & \textbf{p99} & \textbf{N} \\
\midrule
Cloud API (agent\_id, path) & 53\,ms & 63\,ms & 76\,ms & 1,000 \\
Cloud API (agent\_id, body) & 53\,ms & 62\,ms & 77\,ms & 1,000 \\
Cloud API (passport, path) & 54\,ms & 63\,ms & 74\,ms & 1,000 \\
Cloud API (passport, body) & 53\,ms & 63\,ms & 71\,ms & 1,000 \\
Local evaluation & 174\,ms & 243\,ms & 358\,ms & 1,000 \\
\bottomrule
\end{tabular}
\caption{OAP authorization latency benchmarks.}
\label{tab:benchmarks}
\end{table}

Cloud API median is 53\,ms (server-side processing at Cloudflare edge; excludes client-side network round-trip). All four API configurations show consistent p50 (53--54\,ms) and p99 under 77\,ms. Local evaluation (no network) has a median of 174\,ms due to Python subprocess overhead. These latencies are acceptable for agent tool calls, which are themselves I/O-bound (tool execution typically ranges from 200\,ms to several seconds).

\textbf{Breakdown (cloud API, approximate; some steps execute in parallel):}

\begin{table}[H]
\centering
\small
\begin{tabular}{lr}
\toprule
\textbf{Component} & \textbf{Time} \\
\midrule
Passport lookup + cache & 20\,ms \\
Policy evaluation & 15\,ms \\
Decision signing (Ed25519) & 10\,ms \\
HTTP overhead (TLS, parsing, serialization) & ${\sim}$8\,ms \\
\textbf{Synchronous total (blocking)} & \textbf{${\sim}$53\,ms} \\
\bottomrule
\end{tabular}
\caption{Cloud API latency breakdown.}
\label{tab:breakdown}
\end{table}

Audit log write (${\sim}$9\,ms) executes asynchronously after the decision is returned and does not contribute to measured latency.

\subsection{Coverage Analysis: OAP vs.\ OWASP Agentic Top 10}

\begin{table}[H]
\centering
\small
\begin{tabularx}{\columnwidth}{l c X}
\toprule
\textbf{OWASP Risk} & \textbf{Coverage} & \textbf{Mechanism} \\
\midrule
Agent Goal Hijack & Partial & Policy denies actions outside declared capabilities \\
Tool Misuse & Partial & Per-tool capability check + parameter constraints \\
Privilege Escalation & Partial & Capability-scoped passport blocks cross-capability escalation \\
Identity \& Privilege Abuse & Full & Assurance-level gating; Ed25519 verification \\
Improper Output Handling & None & Out of scope (model output, not action) \\
Excessive Agency & Partial & \texttt{max\_calls\_per\_minute}, \texttt{approval\_required} \\
Agentic Supply Chain & Partial & \texttt{allowed\_servers} for MCP; \texttt{allowed\_domains} for web \\
Knowledge Poisoning & None & Out of scope (content, not action) \\
Logging \& Monitoring & Partial & Signed audit log; no anomaly detection \\
Uncontrolled Autonomy & Partial & Fail-closed; ESCALATE not yet implemented \\
\bottomrule
\end{tabularx}
\caption{OAP coverage of OWASP Agentic Top 10 risks.}
\label{tab:owasp}
\end{table}

OAP provides full or partial coverage for 8 of 10 OWASP Agentic risks. The two uncovered risks (Improper Output Handling, Knowledge Poisoning) operate at the content layer, orthogonal to action-authorization.

\subsection{Implementation Status}

\begin{table}[H]
\centering
\small
\begin{tabular}{ll}
\toprule
\textbf{Component} & \textbf{Detail} \\
\midrule
Specification & OAP v1.0, DOI: 10.5281/zenodo.18901596 \\
License & Apache 2.0 \\
Core package & \texttt{@aporthq/aport-agent-guardrails} v1.0.15 \\
Framework adapters & 7 (5 production + 1 beta + 1 proposed) \\
Policy packs & 21 published \\
Signing algorithm & Ed25519 (PKCS8 encoding) \\
Hash chain & SHA-256, per-agent, write-once KV \\
Passport registry & Cloudflare Workers, multi-region \\
Service discovery & \texttt{.well-known/oap/} (RFC 8615) \\
\bottomrule
\end{tabular}
\caption{OAP implementation status.}
\label{tab:implementation}
\end{table}

\section{Standards Alignment}
\label{sec:standards}

OAP aligns with NIST AI RMF~\cite{nist_rmf} functions (GOVERN, MAP, MEASURE, MANAGE) and maps to five NIST SP 800-53 control families (AC, AU, IA, SC, SI). Detailed mappings are in Appendix~\ref{app:nist}. The more significant standards contribution is the proposal of a new control category:

\subsection{Proposed Control Category: Pre-Action Authorization (PAA)}

We propose that standards bodies recognize pre-action authorization as a distinct control category:

\begin{itemize}[nosep]
  \item \textbf{PAA-1}: The organization SHALL define a machine-readable authorization policy specifying which tool calls an AI agent is permitted to execute, under what conditions, at what assurance level.
  \item \textbf{PAA-2}: The organization SHALL implement a platform-level hook that enforces the policy synchronously before each tool call, independent of the model's reasoning.
  \item \textbf{PAA-3}: The organization SHALL issue verifiable credentials (agent passports) binding agents to authorized capability scopes.
  \item \textbf{PAA-4}: The organization SHALL maintain a tamper-evident audit log of all authorization decisions.
  \item \textbf{PAA-5}: In the absence of a valid authorization decision, tool calls SHALL be denied by default.
\end{itemize}

PCAS~\cite{pcas} and this work independently converge on the need for deterministic pre-action enforcement as a distinct architectural layer, differing primarily in policy language and deployment model. This convergence suggests the category is real, not an artifact of a single system's design.

\section{Discussion}

\subsection{Limitations}

\textbf{Delegation chains.} OAP v1.0 does not formalize a delegation model for multi-agent scenarios. An agent that delegates to a sub-agent with narrowed permissions requires a delegation chain specification. This is planned for v1.1.

\textbf{Policy expressiveness.} OAP policies use declarative JSON/YAML, simpler to author and audit but less expressive than PCAS's Datalog-derived language~\cite{pcas}. Complex conditional logic (e.g., ``allow if the caller's manager approved AND the target account is in the same jurisdiction AND the amount is below the 30-day rolling average'') requires either policy composition or extension of the rule language.

\textbf{Adaptive policies.} OAP policies are static; they must be authored and updated manually. AgentGuardian~\cite{agentguardian} demonstrates that policies learned from execution traces can adapt to novel attack patterns. A hybrid approach (static policy baseline + learned anomaly detection) is a natural extension.

\textbf{Scope boundary.} OAP enforces authorization at tool call boundaries. Actions that do not pass through a tool calling interface (direct model output rendering, in-context retrieval, side-channel communication) are out of scope. Sandboxed execution~\cite{nemoclaw,e2b} provides the complementary containment layer for these cases.

\textbf{Composability attacks.} OAP evaluates each tool call independently. A sequence of individually-permitted calls could collectively achieve an unauthorized outcome (e.g., multiple small transfers that exceed an aggregate limit---a ``structuring'' attack). The OAP v1.1 draft includes sliding-window policy packs that maintain per-agent aggregate state to detect and block sequence-based structuring, alongside delegation chain formalization.

\textbf{ESCALATE completeness.} The ESCALATE decision path (Algorithm~1, line~14) is specified but not yet implemented in the reference. Deployments requiring human-in-the-loop approval currently implement ESCALATE at the application layer.

\textbf{Scale.} Current benchmarks are at modest scale (${\sim}$1,000 concurrent evaluations). Performance at 10,000+ evaluations per second, typical of large enterprise deployments, requires dedicated benchmarking.

\textbf{Platform trust.} OAP assumes the framework runtime is not compromised (Section~\ref{sec:formal}.1). If an adversary controls the runtime, the hook can be bypassed entirely. Integration with TEE-based attestation~\cite{proofofguardrail} could reduce this assumption.

\subsection{The Full Safety Stack}

A production agent safety architecture requires at minimum four layers: (1) model alignment, which shapes behavior probabilistically at training time; (2) pre-action authorization, which blocks unauthorized tool calls deterministically before execution; (3) sandboxed execution, which contains the blast radius of actions that pass authorization; and (4) post-hoc evaluation, which finds systemic issues across completed runs and feeds improvements back into alignment and policy.

Remove any one and there is a structural gap. Alignment without authorization is probabilistic enforcement. Authorization without sandboxing lacks blast-radius containment. Sandboxing without authorization permits semantically unauthorized actions within permitted resource scope. Evaluation without runtime enforcement finds problems only after damage is done.

\subsection{From Security to Accountability}

The same infrastructure that enforces security constraints (identity, policy, cryptographic proof) serves a broader accountability function. Policy packs can encode quality gates (minimum test coverage before deployment), operational contracts (SLA requirements before task delegation), and compliance controls (data residency before cross-border transfer). The agent passport is less an access control token and more an accountability credential: a portable record of what an agent is authorized, required, and expected to do.

\section{Conclusion}

Autonomous AI agents executing tool calls without standardized authorization is the access control gap of the current era, analogous to the absence of RBAC in early multi-user systems and the absence of OAuth in early web APIs. The gap is real: 27.2\% of engineering teams building custom authorization from scratch~\cite{gravitee}, 492+ MCP servers deployed without authentication~\cite{invariantlabs}, 824+ malicious skills in the OpenClaw ecosystem~\cite{cisco}, and autonomous agents leaking SSNs through single-verb reframing~\cite{agentsofchaos}. In a live adversarial testbed (4,437 decisions across 1,151 sessions), social engineering succeeded 74.6\% of the time when the model was the primary defense; deterministic policy enforcement reduced this to 0\% (879 highest-tier attempts, \$5,000 prize unclaimed).

Pre-action authorization is distinct from alignment, evaluation, and sandboxing. It complements all three.

OAP is an open-source, specification-first implementation of this layer, deployed across six agent frameworks (with one proposed) and evaluated against live adversarial attacks. Limitations remain: delegation chains are not yet formalized, policy expressiveness is bounded by declarative JSON, benchmarks are at modest scale (${\sim}$1,000 concurrent evaluations; enterprise-scale validation at 10,000+ evaluations/second is needed), and the ESCALATE path is specified but not yet implemented. The specification is released under Apache 2.0 (DOI: 10.5281/zenodo.18901596).

\section*{Disclosure}

The author is the founder of APort Technologies Inc., which develops and operates the OAP reference implementation evaluated in this paper. The specification~\cite{oap_spec} is released under Apache 2.0, all artifacts are publicly available, and any organization can build a conforming implementation. OAP's OpenClaw integration (Section~\ref{sec:oap}) predates the Cisco ClawHavoc research~\cite{cisco}; the Cisco study was conducted independently and no commercial relationship exists between APort and OpenClaw. The Vault CTF dataset will be released (anonymized) alongside the companion empirical paper.

\bibliographystyle{plain}

\appendix

\section{Artifacts and Reproducibility}
\label{app:artifacts}

All artifacts required to reproduce the results in this paper are publicly available under Apache 2.0:

\begin{table}[H]
\centering
\small
\begin{tabularx}{\columnwidth}{l l X}
\toprule
\textbf{Artifact} & \textbf{Location} & \textbf{Purpose} \\
\midrule
OAP v1.0 Spec & DOI: 10.5281/zenodo.18901596 & Canonical spec \\
Spec source & github.com/aporthq/aport-spec & Living spec \\
Policy packs & github.com/aporthq/aport-policies & Policy library \\
Reference impl. & github.com/aporthq/aport-agent-guardrails & Evaluator \\
Adversarial testbed & vault.aport.io & Live CTF \\
API endpoint & aport.io / aport.id & Registry \\
API docs & docs.aport.io & Integration \\
NIST RFI & NIST-2025-0035 (March 2026) & Standards \\
A2A proposal & github.com/a2aproject/A2A/discussions/1404 & Protocol \\
\bottomrule
\end{tabularx}
\caption{Specification artifacts.}
\label{tab:artifacts}
\end{table}

\textbf{To reproduce performance benchmarks:}
\begin{enumerate}[nosep]
  \item Install \texttt{@aporthq/aport-agent-guardrails} v1.0.15 from npm.
  \item Create a passport via \texttt{npx aport-id} or at aport.id.
  \item Run \texttt{POST /api/verify/policy/\{capability\_id\}} with timing instrumentation ($N{=}1{,}000$).
\end{enumerate}

\textbf{To reproduce CTF results:} live results are available at vault.aport.io; the anonymized dataset will be released with the companion empirical paper.

\section{NIST Alignment Detail}
\label{app:nist}

\textbf{NIST AI RMF Mapping} (function-level alignment; sub-category mapping requires organizational context):

\begin{table}[H]
\centering
\small
\begin{tabularx}{\columnwidth}{l X}
\toprule
\textbf{AI RMF Function} & \textbf{OAP Contribution} \\
\midrule
GOVERN & Policy packs encode organizational authorization policy in machine-readable, auditable form \\
MAP & Assurance levels (L0--L4FIN) map agent capabilities to risk categories \\
MEASURE & Signed per-decision audit log enables denial rate, escalation rate, and coverage measurement \\
MANAGE & Policy packs updated without model retraining; passport suspension propagates within 30\,s \\
\bottomrule
\end{tabularx}
\caption{NIST AI RMF mapping.}
\label{tab:nist_rmf}
\end{table}

\textbf{NIST SP 800-53 Alignment:}

\begin{table}[H]
\centering
\small
\begin{tabularx}{\columnwidth}{l X}
\toprule
\textbf{Control Family} & \textbf{OAP Mechanism} \\
\midrule
AC (Access Control) & Policy packs enforce least-privilege per tool type; passports are scoped credentials \\
AU (Audit) & Every tool call decision produces a signed log entry with decision, reason, and timestamp \\
IA (Identification) & Ed25519-signed passport provides cryptographic agent identity \\
SC (System Protection) & \texttt{.well-known/oap/} discovery over HTTPS; Ed25519 decision signatures \\
SI (System Integrity) & Fail-closed policy; hash-chained audit trail with write-once semantics \\
\bottomrule
\end{tabularx}
\caption{NIST SP 800-53 alignment.}
\label{tab:nist_800_53}
\end{table}

\section{Denial Response Handling}
\label{app:denial}

When OAP returns a DENY decision, the response includes:

\begin{lstlisting}[language=json]
{
  "decision": "DENY",
  "deny_code": "oap.limit_exceeded",
  "reason": "Amount $500 exceeds max_per_tx limit of $100",
  "decision_id": "dec_abc123...",
  "signature": "ed25519:..."
}
\end{lstlisting}

Agent frameworks use the structured denial to generate contextual responses rather than generic failures. For example, a banking agent receiving \texttt{oap.merchant\_forbidden} can inform the user that the recipient is not in the approved list, without exposing internal policy details.

\end{document}